\documentclass[12pt]{article}
\usepackage{amsmath}
\usepackage{amssymb}
\usepackage{amsfonts}
\usepackage{latexsym}
\usepackage{color}
\usepackage{graphicx}

\catcode `\@=11 \@addtoreset{equation}{section}

\catcode `\@=12

\newcommand{\be}{\begin{equation}}
\newcommand{\en}{\end{equation}}
\newcommand{\bea}{\begin{eqnarray}}
\newcommand{\ena}{\end{eqnarray}}
\newcommand{\beano}{\begin{eqnarray*}}
\newcommand{\enano}{\end{eqnarray*}}
\newcommand{\bee}{\begin{enumerate}}
\newcommand{\ene}{\end{enumerate}}

\newcommand{\Hil}{{\cal H}}

\newcommand{\Id}{1\!\!1}

\newcommand{\F}{{\cal F}}
\newcommand{\Lc}{{\cal L}}
\newcommand{\Sc}{{\cal S}}
\newcommand{\D}{{\cal D}}
\newcommand{\K}{{\cal K}}

\newcommand{\M}{{\cal M}}

\newcommand{\1}{1 \!\! 1}
\newtheorem{thm}{Theorem}
\newtheorem{cor}[thm]{Corollary}
\newtheorem{lemma}[thm]{Lemma}
\newtheorem{prop}[thm]{Proposition}
\newtheorem{rem}[thm]{Remark}
\newtheorem{defn}[thm]{Definition}

\newenvironment{proof}{\noindent {\bf Proof:}}{\hfill$\Box$}
\begin{document}
\thispagestyle{empty}

\vspace*{1cm}

\begin{center}
{\Large \bf Hamiltonians generated by Parseval frames} 

\vspace{4mm}

{\large F. Bagarello${}^{1,2}$, S. Ku\.{z}el${}^3$}\\
\vspace*{1cm}

\normalsize
${}^1$Dipartimento di Ingegneria, Universit\`a di Palermo, I-90128  Palermo, Italy\\

\vspace*{.5cm}
${}^2$ INFN, Sezione di Napoli, Italy\\

\vspace*{.5cm}
${}^3$ AGH University, Krak\'ow, Poland.\\

\end{center}

\vspace*{0.5cm}

\begin{abstract}
\noindent  It is known that self-adjoint Hamiltonians with purely discrete eigenvalues can be written as  (infinite) linear combination of mutually orthogonal projectors with eigenvalues as coefficients of the expansion. The projectors are defined by the eigenvectors of the Hamiltonians. In some recent papers, this expansion has been extended to the case in which these  eigenvectors form a Riesz basis or, more recently, a $\D$-quasi basis, \cite{bell,bit}, rather than an orthonormal basis. Here we discuss what can be done when these sets are replaced by Parseval frames. 
This interest is motivated by physical reasons, and in particular by the fact that the {\em mathematical } Hilbert space where the physical system is originally defined, contains  sometimes also states which cannot really be occupied by the {\em physical} system itself. In particular,  we show what changes in the spectrum of the observables, when going from orthonormal bases to Parseval frames. 
In this perspective we propose the notion of $E$-connection for observables. Several examples are discussed.
\end{abstract}

\vspace{2cm}


\vfill

\newpage

\section{Introduction}\label{sect1}
As it is well known, in quantum mechanics the dynamics of closed systems is quite often deduced out of a self-adjoint operator, 
the Hamiltonian $H$ of the system $\Sc$, which is the main ingredient to write down the Schr\"odinger equation for the wave function of $\Sc$, $i\dot \Psi(t)=H\Psi(t)$.
 An important aspect of $H$, useful for the analysis of $\Sc$, is the set of its eigenvalues and eigenvectors. Assuming that $H$ has only discrete spectrum, 
 each eigenvalue $E_j$ of $H$ is real, and the eigenvectors corresponding to different eigenvalues are orthogonal. 
 In particular, if each $E_j$ has multiplicity one, the set of related {normalized} eigenvectors, $\F_e=\{e_j\}$,
  is an orthonormal basis {(ONB) of the Hilbert space} $\Hil$, which can be thought as the closure of the linear span of the $e_j$'s. 
  The various {linear combinations of} $e_j$ represent different physical states of the system $\Sc$, e.g., different energy configurations. 
  Of course, the same point of view can be extended to other self-adjoint operators having the same properties required to $H$, or even to other operators with 
  not only discrete spectrum. For instance, if $\hat p$ is the momentum operator defined on $\Lc^2(\mathbb{R})$, its spectrum coincides with $\mathbb{R}$, and each $p\in\mathbb{R}$ corresponds
   to a generalized eigenvectors of $\hat p$, the plane wave $\frac{1}{\sqrt{2\pi}}\,e^{-ipx}$, which does not belong to $\Lc^2(\mathbb{R})$ but can still 
   be normalized in a distributional sense, \cite{mer,mess}. 
  
In the analysis of a concrete physical situation it may happen that \emph{not all} {vectors of $\Hil$} are {\em relevant} in the analysis of $\Sc$. 
This is the case, for instance, when the energy of $\Sc$ cannot {\em really} increase too much, or when $\Sc$ is localized in a bounded region, or still when the value of the momentum 
of $\Sc$ cannot be too large. In all these cases, but not only, it is reasonable to consider a {\em physical vector space}, $\Hil_{ph}$, as the subset of the
 {\em mathematical} Hilbert space $\Hil$ on which $\Sc$ is originally defined. 
 This is exactly the point of view in \cite{jpg}, just to cite one, where $\Hil_{ph}$ contains those functions of  $\Hil=\Lc^2(\mathbb{R})$ 
 which are zero outside a certain compact subset $\D$ of $\mathbb{R}$. Depending on the interpretation of $\Hil$, this approach can be used to describe particles
  localized in $\D$, or particles with a bounded momentum. In \cite{jpg} the main idea is to consider 
  $\Hil_{ph}$ as a subspace of $\Lc^2(\mathbb{R})$ for quantization purposes.
 
  In general, the physical Hilbert space
 $\Hil_{ph}$ can be constructed as the projection of $\Hil$, via some suitable orthogonal projector operator 
 $P$ i.e.,   $\Hil_{ph}=P\Hil$.  
 Then, if $H$ is any (bounded, for the moment) operator on $\Hil$, and $f,g\in\Hil_{ph}$,  its matrix elements $\langle f,Hg\rangle$ can be rewritten as follows
$$
\langle f,Hg\rangle=\langle Pf,HPg\rangle=\langle Pf,(PHP)Pg\rangle=\langle Pf,H_{ph}Pg\rangle,
$$
where,  $H_{ph}=PHP$ is  the {\em physical part of $H$} acting in $\Hil_{ph}$.
Of course, if $H=H^\dagger$, then $H_{ph}=H^\dagger_{ph}$ as well. 

Let the set $\F_e=\{e_n,  \, n=1,2,\ldots, N\}$ of normalized eigenvectors
of $H$ be an ONB of $\Hil$ and let $\{E_n, \, n=1,2,\ldots, N\}$ be the corresponding set of 
real eigenvalues  (here $N=\dim\Hil\leq\infty$). Then the action of the corresponding Hamiltonian $H$ can be presented as
\begin{equation}\label{K1}
Hg=\sum_{n=1}^{N}E_n \langle e_n, g\rangle\,e_n
\end{equation}  
and its maximal  domain of definition is  $D(H)=\{g\in\Hil: \sum_{n=1}^{N}E_n \langle e_n, g\rangle\,e_n\in\Hil\} = \{g\in\Hil: \sum_{n=1}^{N}E_n^2 \,|\langle e_n, g\rangle|^2<\infty\}$. In view of 
\eqref{K1}, the physical part of $H$ has the form
$$
PHPg=\sum_{n=1}^{N}E_n \langle e_n, Pg\rangle\,\varphi_n=\sum_{n=1}^{N}E_n \langle Pe_n, Pg\rangle\,\varphi_n=
\sum_{n=1}^{N}E_n \langle \varphi_n, f\rangle\,\varphi_n,  
$$
where $f=Pg$ and  $\varphi_n=Pe_n.$  Therefore, 
 the operator $H_{ph}=PHP$ acts in a subspace $\Hil_{ph}$ of $\Hil$ and
 \begin{equation}\label{K2}
 H_{ph}f=\sum_{n=1}^{N}E_n \langle \varphi_n, f\rangle\,\varphi_n,  \qquad  f\in{D(H_{ph})}={PD(H)}.
\end{equation} 

In general, the set of vectors $\F_\varphi=\{\varphi_n=Pe_n, \, n=1,2,\ldots, N\}$  loses the property of being an ONB of 
$\Hil_{ph}$ and, instead, it turns out to be a  Parseval frame\footnote{see section \ref{sII.1} for definition and elementary properties of frames} of $\Hil_{ph}$.
For this reason, the operator $H_{ph}$ in \eqref{K2} can be considered as \emph{a Hamiltonian $H_\varphi(=H_{ph})$ generated by a Parseval frame  $\F_\varphi$}, in analogy with what discussed in \cite{bell,bit}.
The operators defined by \eqref{K2} are particular case of more general concept of \emph{multipliers}, which was introduced and studied in \cite{BA1, BA2}.
{Furthermore, for finite Parseval frames $(N<\infty)$, the operators $H_{ph}$
can be regarded as quantum observables that are obtained by using a finite version of the Klauder-Berezin-Toeplitz-type  coherent state quantization \cite{Gaz} of
a real-valued function $f$ defined on a set of data $\{a_1, \ldots,a_N\}$ related to a physical
system. The eigenvalues of $H_{ph}$ form the ‘quantum spectrum’ of a classical observable $f$  whereas its  ‘classical spectrum’ 
coincides with the set of values $\{E_n=f(a_n)\}_{n=1}^N$ \cite{CCGG, CGV}.  These results can be generalized to the case of infinite
Parseval frames with the use of POVM quantization developed in \cite{GH}.}

The main objective of this paper is  the investigation of Hamiltonians generated by Parseval frames.  
As we will see, an interesting aspect is that, when we project from $\Hil$ to $\Hil_{ph}$, the eigenvalues of the observables are not preserved in general. 
This could have relevant consequences in concrete situations, where $H$ is just a formal {\em simple} operator, while $H_{ph}$ is its really useful physical counterpart.
In other words, while each $E_n$ is an eigenvalue for $H$, see \eqref{K1}, $E_n$ is not an eigenvalue for $H_{ph}$ if $\|\varphi_n\|\not=1$, 
despite of the fact that $H$ and $H_{ph}$ share similar expansions \eqref{K1} and \eqref{K2}, but in terms of families of vectors which have different properties.
 In particular, $\F_\varphi$ is not an ONB,
 and therefore $\langle \varphi_j,\varphi_k\rangle\neq \delta_{j,k}$, in general.  
  We construct five explicit examples in Section \ref{sII.2} showing, among other features, how eigenvalues change.

{In the present paper, we concentrate on the case of bounded operators. 
Unbounded Hamiltonians generated by Parseval frames have a lot of delicate properties and they will be considered in a forthcoming paper. }
 
The paper is structured as follows:  after short preliminaries about frames in Section \ref{sII.1}, we begin our analysis of the bounded Hamiltonians  $H_\varphi$.  
General methods of the calculation of eigenvalues and eigenvectors of $H_\varphi$ are presented in  
Theorem \ref{K10} and Corollary \ref{K12}.  
In Section \ref{sII.3},  we introduce and study a physically motivated relation ($E$-connection)
 between Parseval frames and ONBs (Definition \ref{K71}).  It should be mentioned that this relation, when it exists, 
is not related to the possibility of getting an ONB in a larger space (see the Naimark dilation theorem \ref{Naimark},  \cite{hanlar})  
since it works in the same Hilbert space.  In fact, it is more connected with what is stated in \cite[Corollary 8.33]{heil}  or \cite[Theorem 5.5.5]{chri}.
Conclusions are given in Section \ref{secconcl}.

\section{Bounded Hamiltonians defined by Parseval frames}

\subsection{Frames and Parseval frames}\label{sII.1}
Here all necessary information about frame theory are presented in a form convenient for our exposition.
More details on frames can be found in monographs \cite{chri, heil}.
The papers \cite{BCCL, Caz, CazLynch, hanlar}  are recommended as complementary reading on the subject.

Let $\K$ be a complex Hilbert space with scalar product $\langle \cdot, \cdot \rangle$ 
linear in the second argument.  Let $\mathbb{J}$ denote a generic countable (or finite) index set such as
$\mathbb{Z}$, $\mathbb{N}$, $\mathbb{N}\cup\{0\}$, etc. By $|\mathbb{J}|$ we denote the cardinality of
$\mathbb{J}$.

\begin{defn}
A set of vectors $\F_\varphi=\{\varphi_j, \,j\in \mathbb{J}\}$ is called a frame in $\K$ 
if there are constants $A$ and $B$, $0<A\leq B<\infty$, such that, for all $f\in\K$, 
\begin{equation}\label{K3}
A\|f\|^2\leq\sum_{j\in\mathbb{J}}|\langle \varphi_j, f\rangle|^2\leq{B}\|f\|^2.
\end{equation}
\end{defn}
The optimal constants in \eqref{K3}  (maximal for $A$ and minimal for $B$) are
called the frame bounds.  A frame $\F_\varphi$ is called  \emph{a tight frame} if $A = B$, and is
called \emph{a  Parseval frame} (PF in the following) if $A = B = 1$. 

{The potential of a frame 
$\F_\varphi$ is defined by 
${\bf FP}[\F_\varphi]=\sum_{j,i\in\mathbb{J}}|(\varphi_i, \varphi_j)|^2$.
Following \cite{BCCL}, 
we recall that the excess ${\bf e}[\F_\varphi]$ of $\F_\varphi$  is the greatest integer $n$ such that  $n$ elements  
can be deleted from the frame $\F_\varphi$ and still leave a complete set, or $\infty$ if there is no upper bound to the number of elements that can be removed.}

Formula \eqref{K3} with $A = B = 1$ is quite similar to the Parseval equality for ONBs  and in fact it is well known that a PF is indeed an ONB 
if each vector in  $\F_\varphi$ is normalized: $\|\varphi_j\|=1$ for all $j\in\mathbb{J}$, 
and vice-versa.  On the other hand, some of the vectors in a PF may be the zero vector, while this is not allowed for ONBs.    
 Similarly to  ONBs,  PFs have the maximality property:  they cannot be enlarged to a PF by adding non-zero vectors 
 (a unique way to enlarge is to add zero vectors). 

It is well known that, given an ONB $\F_e=\{e_j, \, j\in \mathbb{J}\}$ in a Hilbert space $\Hil$ containing $\K$ as a subspace, and an orthogonal projection 
$P:\Hil\rightarrow\K$, the set $\F_{\varphi}=\{\varphi_j=Pe_j, \, j\in \mathbb{J}\}$ is a PF for $\K$. 
The inverse of this statement is the so-called Naimark  theorem (see, for instance, \cite[Propositions 1.1, 1.4]{hanlar}),  which is crucial for our investigations. More precisely:
\begin{thm}\label{Naimark}
Let $\F_\varphi=\{\varphi_j, \,j\in \mathbb{J}\}$ be a PF in a Hilbert space $\K$. Then there exists a Hilbert space $\M$ and a PF
$\F_\psi=\{\psi_j, \,j\in \mathbb{J}\}$ in $\M$ such that
\begin{equation}\label{K3b}
\F_e=\{e_j=\varphi_j\oplus\psi_j \ :\  j\in \mathbb{J}\}
\end{equation} 
is an ONB for\footnote{If $\F_\varphi$ is an ONB, the understanding is that  $\M$ will be the zero space and each  $\psi_j$ will be zero vector.} $\Hil=\K\oplus\M$.  
\end{thm}
{
\begin{rem}\label{rem2}
It follows from \cite[Lemma 4.1]{BCCL} and the proof of  \cite[Proposition 1.1]{hanlar} that
the dimension of spaces $\K$ and $\M$ in Theorem \ref{Naimark} coincide,  respectively, with the potential ${\bf FP}[\F_\varphi]$ and with the excess 
${\bf e}[\F_{\varphi}]$ of $\F_{\varphi}$. Moreover, $\dim\Hil=|\mathbb{J}|={\bf FP}[\F_\varphi]+{\bf e}[\F_{\varphi}]$.
\end{rem}}

Each Parseval frame  $\F_\varphi$ determines an analysis operator $\theta_\varphi : \K \to \ell_2(\mathbb{J})$:
\begin{equation}\label{K31}
\theta_\varphi(f)=\{\langle \varphi_j, f\rangle\}_{j\in\mathbb{J}}, \qquad f\in\K,
\end{equation}
which is an isometry $\theta_\varphi : \K \to \ell_2(\mathbb{J})$.
The image set ${\mathcal R}(\theta_\varphi)$ is a closed subspace of $\ell_2(\mathbb{J})$ and, because of the Theorem \ref{Naimark},
\begin{equation}\label{K11}
\ell_2(\mathbb{J})={\mathcal R}(\theta_\varphi)\oplus{\mathcal R}(\theta_\psi),
\end{equation}
where the isometry $\theta_\psi : \M \to \ell_2(\mathbb{J})$ is defined by the PF $\F_\psi$ introduced in Theorem \ref{Naimark}.

The next statement is well known in the frame theory. For convenience of the reader  we give its simple proof.
\begin{lemma}\label{K4}
Let $\{c_j\}_{j\in\mathbb{J}}\in\ell_2(\mathbb{J})$. Then $\sum_{j\in\mathbb{J}}{c_j}\varphi_j=0$ if and only if 
$\{c_j\}_{j\in\mathbb{J}}\in{\mathcal R}(\theta_\psi)$.
\end{lemma}
\begin{proof}
If $\sum{c_j}\varphi_j=0$, then $f=\sum{c_j}e_j=\sum{c_j}\varphi_j+\sum{c_j}\psi_j=\sum{c_j}\psi_j$  due to  \eqref{K3b}.
This means that $f\in\M$ and $c_j=\langle e_j, f \rangle=\langle \psi_j, f \rangle$. Hence,  $\theta_\psi(f)=\{c_j\}$.

Conversely, if $\{c_j\}\in{\mathcal R}(\theta_\psi)$, then there exists $f\in\M$ such that $c_j=\langle \psi_j, f \rangle=\langle e_j, f \rangle$. Hence,
$f=\sum{c_j}e_j=\sum{c_j}\varphi_j+\sum{c_j}\psi_j$ that gives  $\sum{c_j}\varphi_j=0$, since $f\in\M$.
\end{proof}

\begin{lemma}\label{K6}
The Gram matrix  $G_\varphi=[\langle \varphi_i, \varphi_j\rangle]_{i,j\in\mathbb{J}}$ of a PF $\F_{\varphi}$  determines a bounded self-adjoint operator in 
$\ell_2(\mathbb{J})$ which is the orthogonal projector in $\ell_2(\mathbb{J})$ onto ${\mathcal R}(\theta_\varphi)$.  
\end{lemma}
\begin{proof}
In view of \cite[Theorem 7.5]{heil},  the Gram matrix $G_\varphi$ determines a  bounded self-adjoint operator in 
$\ell_2(\mathbb{J})$, which acts as follows $G_\varphi : \{c_j\}_{j\in\mathbb{J}} \to \{\sum_{j\in\mathbb{J}}\langle \varphi_i, \varphi_j\rangle{c_j}\}_{i\in\mathbb{J}}$. 
This operator coincides with the orthogonal projection in  $\ell_2(\mathbb{J})$ onto ${\mathcal R}(\theta_\varphi)$ constructed in \cite[Proposition 5.3.6]{chri}
(here, we should take into account that the frame operator $S$ in the formula (5.12) of \cite{chri} is the identity operator, since $\F_{\varphi}$  is a PF and
the scalar product in  \cite{chri} is linear in the first argument).
\end{proof}

\subsection{Working with bounded operators}\label{sII.2}

Assume that $\{E_j\}_{j\in\mathbb{J}}$ is a bounded sequence of real numbers. 
Hence, there exist  $E_{min}=\inf_{j\in\mathbb{J}}{E_j}$ and $E_{max}=\sup_{j\in\mathbb{J}}{E_j}$
such that $-\infty<{E_{min}}{\leq}E_j\leq{E_{max}}<\infty$. 
The Hamiltonian generated by a PF $\F_{\varphi}$ and a collection of numbers $\{E_j\}$:
\begin{equation}\label{K5}
H_{\varphi}f=\sum_{j\in\mathbb{J}}E_j \langle \varphi_j, f\rangle\,\varphi_j,  \qquad  f\in\K
\end{equation}
is a bounded everywhere defined  self-adjoint operator in $\K$.  These properties of $H_{\varphi}$ can be established as  follows: due to Theorem \ref{Naimark} 
there exists an ONB $\{e_j\}$ of $\Hil$ such that $Pe_j=\varphi_j$, where $P$ is the orthogonal projection in $\Hil$ on $\K$.
Therefore,  $H_{\varphi}=PH_{e}P$, where 
\begin{equation}\label{K5b}
H_{e}g=\sum_{j\in\mathbb{J}}E_j \langle e_j, g\rangle\,e_j,  \qquad  g\in\Hil
\end{equation}
is a  bounded  self-adjoint operator in $\Hil$ with eigenvalues $\{E_j\}$ and the corresponding eigenvectors $\{e_j\}$. 

The spectrum of  $H_{e}$ coincides with the closure of  the set of its eigenvalues:  $\sigma(H_e)=\overline{\{E_j\}}$. 
Since $\langle H_{e}f, f\rangle=\langle H_{\varphi}f, f\rangle$ for $f\in\K$, the spectrum $\sigma(H_\varphi)$
is contained in $[E_{min}, E_{max}]$.  In general, as already noticed, we cannot claim that the quantities  $\{E_j\}$ in \eqref{K5} 
are also eigenvalues of $H_{\varphi}$. This will become evident in the examples below.

Let us assume that   $H_{\varphi}f=\mu{f}$ for nonzero $f\in\K$.  Since $\F_{\varphi}$  is a PF, $f=\sum_{j\in\mathbb{J}}\langle \varphi_j, f\rangle\,\varphi_j$
and the eigenvalue equation takes the form:
\begin{equation}\label{K8}
\sum_{j\in\mathbb{J}}(E_j-\mu)\langle \varphi_j, f\rangle\,\varphi_j=0.
\end{equation}
In view of Lemma \ref{K4}, the relation \eqref{K8} is equivalent to the condition  $\{(E_j-\mu)c_j\}\in\mathcal{R}(\theta_\psi)$,
where $\{c_j=\langle \varphi_j, f\rangle\}\in\mathcal{R}(\theta_\varphi)$.  Summing up, this proves:
\begin{thm}\label{K10}
The operator $H_{\varphi}$ has an eigenvalue $\mu$ if and only if there exists a sequence $\{c_j\}\in\mathcal{R}(\theta_\varphi)$ such that 
$\{(E_j-\mu)c_j\}\in\mathcal{R}(\theta_\psi)$, where  $\mathcal{R}(\theta_\varphi)$ and $\mathcal{R}(\theta_\psi)$
are the orthogonal subspaces of $\ell_2(\mathbb{J})$, see \eqref{K11}. In this case, the corresponding eigenvector  is
$f=\sum_{j\in\mathbb{J}}c_j \varphi_j$.
\end{thm}

Lemma \ref{K6} allows one to carry out modifications of Theorem \ref{K10} which, sometimes,  are more convenient for the calculation of
eigenvalues.
Simultaneously with the Gram matrices $G_\varphi=[\langle \varphi_i, \varphi_j\rangle]_{i,j\in\mathbb{J}}$ and $G_\psi=[\langle \psi_i, \psi_j\rangle]_{i,j\in\mathbb{J}}$,
we will use the matrix-valued functions $B_\varphi(\mu)=[b_{ij}(\mu)]_{i,j\in\mathbb{J}}$ and
$D_\varphi(\mu)=[d_{ij}(\mu)]_{i,j\in\mathbb{J}}$, where
\begin{equation}\label{K41}
 b_{ij}(\mu)=\sum_{k\in\mathbb{J}}(E_k-\mu)\langle \varphi_i, \varphi_k\rangle\langle \varphi_k, \varphi_j\rangle, \qquad d_{ij}(\mu)=(E_i-\mu)\langle \varphi_i, \varphi_j\rangle.
\end{equation}

\begin{cor}\label{K12}
Let  $H_{\varphi}$ be a Hamiltonian generated by a PF $\F_{\varphi}$, see \eqref{K5}. The following are equivalent:
\begin{enumerate}
\item[(i)] $\mu\in\mathbb{R}$ is an eigenvalue of $H_{\varphi}$;
\item[(ii)] there exists $\{c_j\}\in\ell_2(\mathbb{J})$ such that $G_\varphi\{c_j\}\not=0$ and $B_\varphi(\mu)\{c_j\}=0$;
\item[(iii)] there exist  $\{c_j\},   \{\tilde{c}_j\}\in\ell_2(\mathbb{J})$ such that $G_\varphi\{c_j\}\not=0$ and $D_\varphi(\mu)\{c_j\}=G_\psi\{\tilde{c}_j\}$.
\end{enumerate} 
\end{cor}
\begin{proof}
Let $\F_e=\{e_j=\varphi_j\oplus\psi_j \ :\  j\in \mathbb{J}\}$ be an ONB of $\Hil$ determined in Theorem \ref{Naimark} and let
\begin{equation}\label{K31b}
\theta_e(g)=\{\langle e_j, g\rangle\}_{j\in\mathbb{J}}, \qquad g\in\Hil
\end{equation}
be the corresponding  analysis operator that isometrically maps $\Hil$ onto $\ell_2(\mathbb{J})$.
In view of \eqref{K5b},
$\theta_e(H_eg)=\mathcal{E}\theta_e(g)$, 
where  $\mathcal{E}$ is a bounded multiplication operator in $\ell^2(\mathbb{J})$:
 \begin{equation}\label{K201}
 \mathcal{E}\{c_j\}=\{E_jc_j\}, \qquad  {D}(\mathcal{E})=\ell^2(\mathbb{J}).
 \end{equation}
 
Since, for all $f\in\K$,   $\theta_e(f)=\{\langle e_j, f\rangle=\langle Pe_j, f\rangle=\langle \varphi_j, f\rangle\}=\theta_\varphi(f)$, 
where $P$ is the orthogonal projection in $\Hil$ on $\K$ and $\theta_e$ is defined by
\eqref{K31b},  we  arrive at the conclusion that $\theta_e$ maps  $\K$ onto  $\mathcal{R}(\theta_\varphi)$
(see the decomposition \eqref{K11}).
Taking into account that $H_\varphi=PH_eP$,  we obtain that  $H_\varphi$ is unitary equivalent to the operator $\mathcal{E}_\varphi$ 
  acting in the subspace $\mathcal{R}(\theta_\varphi)$ of $\ell^2(\mathbb{J})$:
  \begin{equation}\label{K101}
  \mathcal{E}_\varphi=\mathcal{P}\mathcal{E}\mathcal{P},  \qquad  D(\mathcal{E}_\varphi)=\mathcal{R}(\theta_\varphi)=\mathcal{P}\ell^2(\mathbb{J}),
  \end{equation}
 where $\mathcal{P}=\theta_e{P}\theta_e^{-1}$ is the orthogonal projection operator in $\ell^2(\mathbb{J})$ on $\mathcal{R}(\theta_\varphi)$.

$(i)\iff(ii)$. In view of \eqref{K101},  $\mu\in\sigma_p(H_\varphi)$ if and only if  there exists $\{c_j\}\in\ell_2(\mathbb{J})$ such that
$\mathcal{P}\{c_j\}\not=0$ and $\mathcal{P}(\mathcal{E}-\mu{I})\mathcal{P}\{c_j\}=0$. 
  By Lemma \ref{K6}, the operator  $\mathcal{P}$ coincides with the Gram operator $G_\varphi$. Therefore, the first condition takes
  the  form $G_\varphi\{c_j\}\not=0$,  whereas  the second one: $G_\varphi(\mathcal{E}-\mu{I})G_\varphi\{c_j\}=B_\varphi(\mu)\{c_j\}=0$.

To prove the equivalence of  $(i)$ and $(iii)$ it suffices to repeat the previous arguments and rewrite the condition  $\mathcal{P}(\mathcal{E}-\mu{I})\mathcal{P}\{c_j\}=0$
as $(\mathcal{E}-\mu{I})\mathcal{P}\{c_j\}\in\mathcal{R}(\theta_\psi)$. Taking into account that $G_\psi$ is an orthogonal projection on $\mathcal{R}(\theta_\psi)$
and  $(\mathcal{E}-\mu{I})\mathcal{P}\{c_j\}=D_\varphi(\mu)\{c_j\}$ we complete the proof.
\end{proof}

\subsection{Examples}\label{sectexe}

In this section we propose some examples of $H$ and their related $H_{ph}$, starting with purely mathematical examples, and then considering applications
 arising in concrete applications discussed in the literature. In particular, in Examples 1, 2, and 3, we will show how a physical Hamiltonian $H_{\varphi}=\sum_{j=1}^N E_j \langle \varphi_j, \cdot\rangle\,\varphi_j$, given in terms of some particular PF, can be rewritten in terms of an ONB of its eigenvectors, and how its related eigenvalues are different from the $E_j$ in the expansion above. Of course, in view of what we have discussed in the Introduction, $H_{\varphi}$ should be understood as the physical part of another, {\em larger}, Hamiltonian, $H$, which produces $H_{ph}$ after a suitable projection. This means that we have a first ONB, given by the eigenvectors of $H$, which, when projected to $\Hil_{ph}$, defines a PF. Hence, to find the physical containt of $H_{ph}$, we need to diagonalize $H_{ph}$, getting a second ONB, different from the first one, which can really be considered as the {\em physical eigenvectors} of the system. In Examples 4 and 5 these steps will be particularly emphasized.

{\bf Example 1. Mercedes frame.--}
In  the Hilbert space $\Hil=\mathbb{C}^3$, we consider the ONB
$\F_e=\{e_j,  j\in\mathbb{J}\}$, where $\mathbb{J}=\{1,2,3\}$ and
$$
e_1=\sqrt{\frac{2}{3}}\left[
\begin{array}{c}
1 \\
0 \\
\frac{1}{\sqrt{2}}\\
\end{array}
\right], \qquad e_2=\sqrt{\frac{2}{3}}\left[
\begin{array}{c}
-\frac{1}{2} \\
\frac{\sqrt{3}}{2} \\
\frac{1}{\sqrt{2}} 
\end{array}
\right], \qquad e_3=\sqrt{\frac{2}{3}}\left[
\begin{array}{c}
-\frac{1}{2} \\
-\frac{\sqrt{3}}{2} \\
\frac{1}{\sqrt{2}}
\end{array}
\right].
$$
The orthogonal projection of $\F_e=\{e_1, e_2, e_3\}$ onto the subspace $\K=\mathbb{C}\oplus\mathbb{C}\oplus{0}$
gives rise to a PF $\F_\varphi=\{\varphi_j, j\in\mathbb{J}\}$, where 
$$
\varphi_1=\sqrt{\frac{2}{3}}\left[
\begin{array}{c}
1 \\
0 \\
0
\end{array}
\right], \qquad \varphi_2=\sqrt{\frac{2}{3}}\left[
\begin{array}{c}
-\frac{1}{2} \\
\frac{\sqrt{3}}{2} \\
0
\end{array}
\right], \qquad \varphi_3=\sqrt{\frac{2}{3}}\left[
\begin{array}{c}
-\frac{1}{2} \\
 -\frac{\sqrt{3}}{2}\\
 0
\end{array}
\right].
$$
The PF $\F_\varphi$ is called the Mercedes frame. Its dual PF $\F_\psi=\{\psi_j, j\in\mathbb{J} \}$ in \eqref{K3} consists of vectors  
of the subspace $\M={0}\oplus{0}\oplus\mathbb{C}$ (see formula \eqref{K3b})
$$
\psi_1=\psi_2=\psi_3=\sqrt{\frac{2}{3}}\left[
\begin{array}{c}
0 \\
 0\\
 \frac{1}{\sqrt{2}}
\end{array}
\right]=\left[
\begin{array}{c}
0 \\
 0\\
 \frac{1}{\sqrt{3}}
\end{array}
\right].
$$

 The image sets  of the analysis operators $\theta_\varphi$ and $\theta_\psi$ 
are subspaces of $\ell_2(\mathbb{J})=\mathbb{C}^3$:
$$
\mathcal{R}(\theta_\varphi)=\left\{\left[\begin{array}{c}
2x_1 \\
 -x_1+\sqrt{3}x_2 \\
 -x_1+\sqrt{3}x_2
 \end{array}\right], \ x_1, x_2\in\mathbb{C}\right\}, \qquad \mathcal{R}(\theta_\psi)=\left\{\left[\begin{array}{c}
x_3 \\
x_3 \\
 x_3
 \end{array}\right], \ x_3\in\mathbb{C}\right\}.
$$ 

Let  $E_1$, $E_2$, and $E_3$ be  real quantities. Then the Hamiltonian\footnote{We are assuming here that $H_\varphi$ arises from a different operator $H$, acting on a larger Hilbert space $\mathbb{C}^3$, after taking its projection $H_\varphi=PHP$ on the physical Hilbert space $\Hil_{ph}=\mathbb{C}^2$. Notice that we are not giving here the explicit form of $H$, since it is not really relevant. We will do this in Example 4 and Example 5, since in those cases $H$ has a physical meaning.}
$H_{\varphi}=\sum_{j=1}^3E_j \langle \varphi_j, \cdot\rangle\,\varphi_j$
generated by the Mercedes frame acts in $\mathbb{C}^2$ (we identify $\K$ with $\mathbb{C}^2$) and, due to  Theorem \ref{K10}, 
$\mu\in\mathbb{R}$ is an eigenvalue of $H_{\varphi}$ if and only if the linear system
$$
\left\{\begin{array}{l}
2(E_1-\mu)x_1=x_3 \\
(E_2-\mu)(-x_1+\sqrt{3}x_2)=x_3 \\
(E_3-\mu)(-x_1-\sqrt{3}x_2)=x_3
\end{array}\right.
$$
has a nonzero solution $x_1, x_2, x_3$.  An elementary calculation shows that $\mu\in\sigma(H_{\varphi})$
if and only if
\begin{equation}\label{K21} 
(E_1-\mu)(E_2-\mu)+(E_2-\mu)(E_3-\mu)+(E_1-\mu)(E_3-\mu)=0. 
\end{equation}
Another way to obtain \eqref{K21} is, of course, to present $H_{\varphi}$  (acting in $\mathbb{C}^2$) in the matrix form
\begin{equation}\label{K16}
H_{\varphi}=\frac{1}{6}\left[
\begin{array}{cc}
 4E_1+E_2+E_3 & \sqrt{3}(E_3-E_2)\\
\sqrt{3}(E_3-E_2) &  3(E_2+E_3) \\
\end{array}
\right]
\end{equation}
and, then, to solve the characteristic equation $\det[H_{\varphi}-\mu{I}]=0$.  More detailed calculations show that  the eigenvalues
 are 
 $$
 \tilde E_{\pm}=\frac{1}{3}\left(E_1+E_2+E_3\pm \sqrt{E_1^2+E_2^2+E_3^2-E_1E_2-E_2E_3-E_1E_3}\right), 
 $$ with eigenvectors (if $E_2\neq E_3$)
$$
\tilde e_\pm=N_\pm\left[\begin{array}{c}
-2E_1+E_2+E_3\mp 2\sqrt{E_1^2+E_2^2+E_3^2-E_1E_2-E_2E_3-E_1E_3} \\
 \sqrt{3}(E_2-E_3)
 \end{array}\right].
$$
Here $N_\pm$ are normalization constants fixed to have $\|\tilde e_\pm\|=1$. This means that, using a bra-ket notation,
the operator $H_{\varphi}$ can be rewritten in term of an ONB of its eigenstates:
$$
H_{\varphi}=\sum_{j=1}^{3}\,E_j|\varphi_j \rangle\langle\varphi_j|=\sum_{i=\pm}\,\tilde E_i|\tilde e_i\rangle\langle\tilde e_i|.
$$
This equality clarifies once more that, while $\tilde E_i$ is an eigenvalue of $H_\varphi$, $E_j$ is not.

{\bf Example 2.--}
Let $\{\mathsf{e}_1, \mathsf{e}_2, \ldots \mathsf{e}_K\}$, $K\geq{2}$ be an ONB of a Hilbert space $\K$. The set 
$\F_\varphi=\{\varphi_j, j\in\mathbb{J}\}$ where $\mathbb{J}=\{1, 2, \ldots, K+1\}$ and
\begin{equation}\label{K61}
\varphi_j=\mathsf{e}_j-\frac{1}{K}\sum_{i=1}^{K}\mathsf{e}_i,  \quad 1\leq{j}\leq{K}, \qquad \varphi_{K+1}=\frac{1}{\sqrt{K}}\sum_{i=1}^{K}\mathsf{e}_i,
\end{equation}
is a PF for $\K$ \cite[Lemma 2.5]{CC98}.  Representing each vector $f\in\K$ as $f=\sum_{i=1}^Kx_i\mathsf{e}_i$ and using \eqref{K31b}, 
we define the subspaces $\mathcal{R}(\theta_\varphi)$, $\mathcal{R}(\theta_\psi)$ of $\ell_2(\mathbb{J})=\mathbb{C}^{K+1}$:
$$
\mathcal{R}(\theta_\varphi)=\left\{\left[\begin{array}{c}
x_1 -\frac{1}{K}\sum_{i=1}^kx_i \\
 \vdots \\
 x_k-\frac{1}{K}\sum_{i=1}^kx_i \\
 \frac{1}{\sqrt{K}}\sum_{i=1}^kx_i
 \end{array}\right], \ x_i \in\mathbb{C}\right\}, \quad \mathcal{R}(\theta_\psi)=\mathbb{C}^{K+1}\ominus\mathcal{R}(\theta_\varphi)=\left\{\left[\begin{array}{c}
x \\
\vdots \\
 x \\
 0
 \end{array}\right], \ x\in\mathbb{C}\right\}.
$$ 

Let  $E_1<E_2<\ldots<E_{K+1}$ be  real quantities. Then the Hamiltonian $H_{\varphi}=\sum_{j=1}^{K+1}\,E_j|\varphi_j \rangle\langle\varphi_j|$
generated by the frame $\F_\varphi$ acts in $\K$ and, due to  Theorem \ref{K10}, its  eigenvalues $\tilde{E}_j$ coincide with 
$\mu\in\mathbb{R}$ for which the linear system
$$
\left\{\begin{array}{l}
(E_1-\mu)(x_1 -\frac{1}{K}\sum_{i=1}^kx_i)=x \\
\vdots \\
(E_K-\mu)(x_K -\frac{1}{K}\sum_{i=1}^kx_i)=x \vspace{2mm} \\
(E_{K+1}-\mu)(\frac{1}{\sqrt{K}}\sum_{i=1}^kx_i)=0
\end{array}\right.
$$
has a nonzero solution $x_1, x_2, \ldots x_K, x$. An elementary analysis shows that the largest eigenvalue $\tilde{E}_K$ of $H_{\varphi}$ 
coincides with $E_{K+1}$, while the other eigenvalues $\tilde{E}_j$, $j=1,\ldots{K-1}$ are the roots of the equation
\begin{equation}\label{K33} 
\frac{1}{E_1-\mu}+\frac{1}{E_2-\mu}+\ldots+\frac{1}{E_{K}-\mu}=0.
\end{equation}
The corresponding eigenfunctions are 
$$
f_K=\frac{1}{K}\sum_{i=1}^K\mathsf{e}_i  \quad  (\mbox{for} \ \tilde{E}_K=E_{K+1})  \quad \mbox{and} \quad
f_j=\sum_{i=1}^K\frac{1}{E_i-\mu_j}\mathsf{e}_i,
$$ 
where $\mu_j$ ($j=1,\ldots{K-1}$) is a solution of \eqref{K33}. After the normalization of
$f_j$ we obtain $H_{\varphi}=\sum_{j=1}^{K}\,\tilde{E}_j|f_j \rangle\langle{f}_j|$, which is a different way to write $H_{\varphi}=\sum_{j=1}^{K+1}\,E_j|\varphi_j \rangle\langle\varphi_j|$.

\vspace{2mm}

{\bf Example 3.--} The finite PF considered above is a key counterpart of an infinite PF which 
contains no subset that is a Riesz basis \cite{CC98}.  We investigate spectral properties of Hamiltonians generated by this curious PF. 

Let $\K$ be a separable Hilbert space. Index an ONB for $\K$ as $\{\mathsf{e}_j^K\}_{K\in\mathbb{N}, j=1,\ldots,K}$.
Set $\K_K=\mbox{span}\{\mathsf{e}_1^K, \mathsf{e}_2^K, \ldots \mathsf{e}_K^K\}$.  The vectors $\varphi_j^K\equiv\varphi_j$ defined by
\eqref{K61} form a PF $\F_{\varphi^K}=\{\varphi_j^K, \ j=1, 2\ldots{K+1}\}$ of the space $\K_K$.
Since $\K=\sum_{K=1}^\infty\oplus\K_K$, the collection of vectors $\F_{\varphi}=\sum_{K=1}^\infty\oplus\F_{\varphi^K}$
is a PF for $\K$  (and it was constructed in  \cite{CC98}). 

Assume that $\{E_j^{K}\}_{K\in\mathbb{N}, j=1,\ldots,K+1}$ is a bounded strictly increasing sequence, i.e. 
$E_1^1<E_2^1<E_1^2<E_{2}^2<E_3^2\ldots<E_1^K<E_2^K\ldots<E_K^K<E_{K+1}^K\ldots$.
The Hamiltonian generated by $\F_{\varphi}$ is 
$$
H_\varphi=\sum_{K=1}^{\infty}\sum_{j=1}^{K+1}\,E_j^K|\varphi_j^K \rangle\langle\varphi_j^K|
$$
It is easy to verify (using the previous example) that the point spectrum of $H_\varphi$ involves the subset
$\{E_{K+1}^K\}_{K=1}^\infty$ of the original quantities and the solutions of the equations, cf. \eqref{K33}: 
$$
\frac{1}{E_1^K-\mu}+\frac{1}{E_2^K-\mu}+\ldots+\frac{1}{E_{K}^K-\mu}=0, \qquad K\geq{2}.
$$  
\vspace{2mm}

Let us now discuss a fourth, physically motivated, example, based on the anti-commutation relations (CAR) for two fermionic modes.

{\bf Example 4.--} Let $a_1$ and $a_2$ be the operators satisfying the CAR $\{a_j,a_k^\dagger\}=\delta_{j,k}\1$, with $\{a_j,a_k\}=0$, $j,k=1,2$. 
Here $\1$ is the identity operator in the Hilbert space of the system, which is $\Hil=\mathbb{C}^4$. Calling $\eta_{0,0}$ the vacuum of $a_j$, $a_j\eta_{0,0}=0$, $j=1,2$, we can construct three more vectors acting on it with $a_j^\dagger$: $\eta_{1,0}=a_1^\dagger \eta_{0,0}$, $\eta_{0,1}=a_2^\dagger \eta_{0,0}$ and $\eta_{1,1}=a_1^\dagger a_2^\dagger \eta_{0,0}$. An explicit realization of these vectors and operators
 is the following: $\eta_{0,0}=\delta_1$, $\eta_{1,0}=\delta_2$, $\eta_{0,1}=\delta_3$, $\eta_{1,1}=\delta_4$, where $\{\delta_j\}$ is the canonical ONB in $\mathbb{C}^4$, and
$$
a_1=\left[
\begin{array}{cccc}
0 & 1& 0 & 0 \\
0 & 0& 0 & 0 \\
0 & 0& 0 & 1 \\
0 & 0& 0 & 0 \\
\end{array}
\right], \qquad a_2=\left[
\begin{array}{cccc}
0 & 0& 1 & 0 \\
0 & 0& 0 & -1 \\
0 & 0& 0 & 0 \\
0 & 0& 0 & 0 \\
\end{array}
\right].
$$ 
In \cite{bagoli} these operators have been used to construct a dynamical system describing two populations moving in a two-dimensional lattice, 
and interacting adopting a sort of predator-prey mechanism. The dynamics has been produced by the sum of several copies of a single-cell term, $H_\alpha$ ($\alpha$ labels the lattice cells), 
plus a global contribution responsible for the migration of the species. Here we only consider the single-cell term, which we rewrite as follows:
$$
H=\omega_1 a_1^\dagger a_1+\omega_2 a_2^\dagger a_2 +\lambda(a_1^\dagger a_2+a_2^\dagger a_1),
$$
 where $\omega_1$, $\omega_2$ and $\lambda$ are parameters whose values can be fixed in different way, according to which aspect of the system we want to put in evidence, see \cite{bagoli}.
We also refer to \cite{bagoli} for the meaning of this Hamiltonian, for the rationale for its introduction, and for the dynamics connected to it. 
The matrix expression for $H$ is
$$
H=\left[
\begin{array}{cccc}
0 & 0& 0 & 0 \\
0 & \omega_1& \lambda & 0 \\
0 & \lambda& \omega_2 & 0 \\
0 & 0& 0 & \omega_1+\omega_2 \\
\end{array}
\right],
$$
whose eigenvalues are
$$
E_1=0, \quad E_2=\omega_1+\omega_2 \quad  E_3=\frac{1}{2}\left(\omega_1+\omega_2-\sqrt{(\omega_1-\omega_2)^2+4\lambda^2}\right)
$$ 
and 
$$E_4=\frac{1}{2}\left(\omega_1+\omega_2+\sqrt{(\omega_1-\omega_2)^2+4\lambda^2}\right)
$$
with (normalized) eigenvectors
$$
e_1=\left[
\begin{array}{c}
1 \\
0 \\
0\\
0\\
\end{array}
\right], \qquad e_2=\left[
\begin{array}{c}
0 \\
0 \\
0\\
1\\
\end{array}
\right], \quad \mbox{ and } \quad e_3=\frac{f_3}{\|f_3\|}, \qquad e_4=\frac{f_4}{\|f_4\|},
$$
where
$$
f_3=\left[
\begin{array}{c}
0 \\
\frac{1}{2\lambda}\left(\omega_1-\omega_2-\sqrt{(\omega_1-\omega_2)^2+4\lambda^2}\right) \\
1\\
0\\
\end{array}
\right], \quad f_4=\left[
\begin{array}{c}
0 \\
\frac{1}{2\lambda}\left(\omega_1-\omega_2+\sqrt{(\omega_1-\omega_2)^2+4\lambda^2}\right) \\
1\\
0\\
\end{array}
\right].
$$
For concreteness, if we take $\lambda=2$, $\omega_1=\frac{1}{2}$ and $\omega_2=\frac{7}{2}$, then
$$
H=\left[
\begin{array}{cccc}
0 & 0& 0 & 0 \\
0 & 1/2& 2 & 0 \\
0 & 2& 7/2 & 0 \\
0 & 0& 0 & 4 \\
\end{array}
\right], \qquad e_3=\frac{1}{\sqrt{5}}\left[
\begin{array}{c}
0 \\
-2 \\
1\\
0\\
\end{array}
\right], \quad e_4=\sqrt{\frac{4}{5}}\left[
\begin{array}{c}
0 \\
\frac{1}{2} \\
1\\
0\\
\end{array}
\right],
$$
while $e_1$ and $e_2$ are those previously introduced. Also, $E_1=0$, $E_2=4$, $E_3=-\frac{1}{2}$, $E_4=\frac{9}{2}$,  and we can rewrite 
$H=\sum_{j=1}^4\, E_j|e_j\rangle\langle{e}_j|$.  

Due to the physical interpretation of the model, \cite{bagoli}, if we consider the orthogonal projector $P=\1-P_{1,0}$, $P_{1,0}f=\langle \eta_{1,0},f\rangle\eta_{1,0}$, $P$
project the system on a space in which it is impossible to find the system in a state with high density of the first species, and low density of the other species.
In other words, this state is forbidden for us. The (biological) reason for requiring this is that, for instance, we want to avoid a dominance of the first species on the second one. 
The computation of $H_{ph}$ is easy:  as for the effect of $P$ on the eigenvectors  $\{e_j\}$ of $H$, we get four vectors $\varphi_j=Pe_j$ in $\Hil=\mathbb{C}^4$.
 Removing the zero second component  from each one of these vectors\footnote{This component must be zero due to the action of $P$ on $e_j$.}
 we recover the following  PF $\F_\varphi=\{\varphi_j, 1\leq{j}\leq{4}\}$ in $\K=\mathbb{C}^3$,
\begin{equation}\label{K51}
\varphi_1=\left[
\begin{array}{c}
1 \\
0 \\
0\\
\end{array}
\right], \qquad \varphi_2=\left[
\begin{array}{c}
0 \\
0 \\
1\\
\end{array}
\right], \qquad \varphi_3=\cos\beta\left[
\begin{array}{c}
0 \\
1 \\
0\\
\end{array}
\right], \qquad \varphi_4=\sin\beta\left[
\begin{array}{c}
0 \\
1 \\
0\\
\end{array}
\right],
\end{equation}
where 
$$
\cos\beta=\frac{1}{\|f_3\|}=\frac{2|\lambda|}{\sqrt{[\omega_1-\omega_2-\sqrt{(\omega_1-\omega_2)^2+4\lambda^2}]^2+4\lambda^2}},
$$
$$
 \sin\beta=\frac{1}{\|f_4\|}=\frac{2|\lambda|}{\sqrt{[\omega_1-\omega_2+\sqrt{(\omega_1-\omega_2)^2+4\lambda^2}]^2+4\lambda^2}}, \quad \beta\in(0, {\pi}/{2}).
$$

The  operator  $H_{ph}$ acting in $\K=\mathbb{C}^3$ coincides with the operator $H_\varphi=\sum_{j=1}^4\, E_j|\varphi_j\rangle\langle{\varphi}_j|$
generated by $\F_\varphi$.  In view of \eqref{K41} and \eqref{K51}, the Gram matrix $G_\varphi$ and $B_\varphi(\mu)$  are:
$$
G_\varphi=\left[\begin{array}{cccc}
1 & 0 &  0 & 0 \\
0  & 1 &  0 & 0  \\
0  & 0 & \cos^2\beta & \cos\beta\sin\beta \\
0  & 0 & \cos\beta\sin\beta &  \sin^2\beta
\end{array}\right],  \qquad  B_\varphi(\mu)=\left[\begin{array}{cccc}
E_1-\mu & 0 &  0 & 0 \\
0 & E_2-\mu & 0 & 0   \\
0  &  0  &  b_{33}  & b_{34} \\
0  &  0  &  b_{43}  & b_{44}
\end{array}\right],
$$
where $b_{33}=(E_3-\mu)\cos^4\beta+(E_4-\mu)\cos^2\beta\sin^2\beta$, $b_{44}=(E_3-\mu)\cos^2\beta\sin^2\beta+(E_4-\mu)\sin^4\beta$,
$b_{34}=b_{43}=(E_3-\mu)\cos^3\beta\sin\beta+(E_4-\mu)\cos\beta\sin^3\beta$.

These matrices and  Corollary \ref{K12} allow one to describe eigenvalues $\tilde{E}_1$, $\tilde{E_2}$, and $\tilde{E_3}$
of $H_{ph}=H_\varphi$. Precisely, $\tilde{E}_1=E_1=0$,   $\tilde{E}_2=E_2=\omega_1+\omega_2$,  while $\tilde{E_3}$ coincides with the solution
$\mu$ of the equation 
$(E_3-\mu)\cos^2\beta+(E_4-\mu)\sin^2\beta=0.$ So we see how the eigenvalues of $H_{ph}$ are different from those of the original Hamiltonian $H$. 
Of course, these differences have consequences on the dynamics of the system, but this aspect will not be discussed in this paper.
However, we want to stress that differences in the eigenvalues imply, among other differences, different stationary states. This might have serious, and very interesting consequences: a system which, in principle, should not evolve in time (since it is in a supposed stationary state), does indeed evolve. Or vice-versa. In both cases this can be related to the fact that the {\em real} 
Hamiltonian for the system is not $H$, but $H_{ph}$. This is because of some constraint on the system.

\vspace{2mm}

The last example we want to discuss here is again built in terms of  anti-commutation relations (CAR) 
and fermionic modes, and is based on an ecological system considered first in \cite{bagoli2} and then in \cite{bagbook2}.

{\bf Example 5.--} The system we want to describe is made by two levels of organisms ($L_1$ and $L_2$),  one compartment for the nutrients 
and another compartment for the garbage, see Figure \ref{fig01} for a schematic view. The nutrients feed the organisms of $L_1$, which feed those of $L_2$. Moreover, when dying, the organisms of both levels contribute to increase the density of the garbage which, after some time, turns into nutrients. This is a simple example of closed ecosystem. More complicated systems (with more levels and with different kind of garbages) are considered in \cite{bagoli2}.

\begin{figure}
		
	\begin{picture}(450,90)
	\put(45,35){\thicklines\line(1,0){80}}
	\put(45,55){\thicklines\line(1,0){80}}
	\put(45,35){\thicklines\line(0,1){20}}
	\put(125,35){\thicklines\line(0,1){20}}
	\put(85,45){\makebox(0,0){$L_1$}}
	
	\put(305,-5){\thicklines\line(1,0){120}}
	\put(305,55){\thicklines\line(1,0){120}}
	\put(305,55){\thicklines\line(0,-1){60}}
	\put(425,55){\thicklines\line(0,-1){60}}
	\put(365,20){\makebox(0,0){Garbage}}
			
	\put(45,-5){\thicklines\line(1,0){80}}
	\put(45,15){\thicklines\line(1,0){80}}
	\put(45,-5){\thicklines\line(0,1){20}}
	\put(125,-5){\thicklines\line(0,1){20}}
	\put(85,5){\makebox(0,0){$L_2$}}

	\put(305,-55){\thicklines\line(1,0){120}}
	\put(305,-35){\thicklines\line(1,0){120}}
	\put(305,-55){\thicklines\line(0,1){20}}
	\put(425,-55){\thicklines\line(0,1){20}}
	\put(365,-45){\makebox(0,0){Nutrients}}
		
	\put(370,-10){\vector(0,-1){20}}
	\put(85,33){\vector(0,-1){16}}	
	
	\put(130,44){\vector(1,0){170}}
	\put(130,4){\vector(1,0){170}}
	\put(300,-45){\line(-1,0){280}}
	\put(20,-45){\line(0,1){90}}
	\put(20,45){\vector(1,0){20}}
		
	\end{picture}
	
	\vspace{2cm}
	\caption{\label{fig01}\footnotesize A schematic view to the single-garbage ecosystem.}
\end{figure}
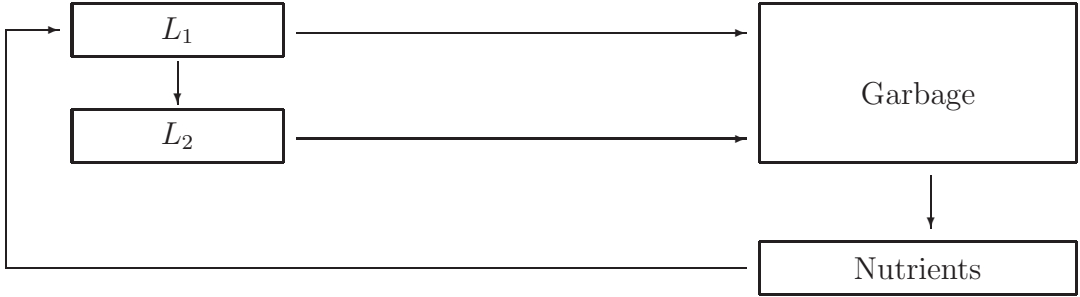

The dynamics of our system is described by  the following Hamiltonian:
$$
\left\{
\begin{aligned}
	&H=H_0+ H_I, \qquad \hbox{ with }  \\
	&H_0 = \sum_{j=0}^{3}\,\omega_{j}\, a_j^\dagger\,a_j,\\
	&H_I = \sum_{j=0}^2\,\lambda_j\left(a_j\,a_{3}^\dagger+a_{3}a_j^\dagger\right)+
	\sum_{j=0}^{1}\,\nu_j\left(a_j\,a_{j+1}^\dagger+a_{j+1}a_j^\dagger\right),
\end{aligned}
\right. 
$$
where $\{a_j,a_k^\dagger\}=\delta_{j,k}\Id$, and  $a_j^2=0$, for all $j, k=0,1,2,3$, and where
$\omega_j, \nu_j$ and $\lambda_j$ are real constants, whose meaning is explained in \cite{bagbook2, bagoli2}. The zero-th fermionic mode ($j=k=0$) is related to the nutrients, the $3$-th mode to the garbage, while
 the two remaining modes describe the organisms of the various trophic levels. We again refer to \cite{bagoli2} for the biological meaning of the various terms in $H$. Here, we just comment that, for instance, the term $\lambda_j a_{3}a_0^\dagger$, describes the fact that the garbage is
recycled by decomposers and transformed into nutrients: this is due to the fact that the density of the garbage decrease (because of the presence of the lowering operator $a_3$) , and simultaneously the density of the nutrients increases (because of the raising operator $a_0^\dagger$).

The equations of motion can be  deduced using the Heisenberg rule $\dot X=i[H,X]$, and the solution (which is not essential here) can be found in \cite{bagoli2}. 
Here we want to see how the idea introduced in this paper works for this particular example. For that, let $\mathsf{e}_{0,0,0,0}$ be the vacuum of the system: 
$a_j\mathsf{e}_{0,0,0,0}=0$ for $j=0,1,2,3$. This corresponds to an {\em essentially empty} system: very low densities in all the levels (including garbage and nutrients). 
Analogously, $\mathsf{e}_{0,0,1,0}$ is a vector describing a situation where only $L_2$ is filled, while all the other levels are essentially empty, and so on. 
We now fix, for concreteness, $\omega_0=\omega_3=2$, $\omega_1=3$, $\omega_2=4$, $\lambda_0=1$, $\lambda_1=2$, $\lambda_2=3$, $\nu_0=1$ and $\nu_1=3$. The criterium for choosing some particular values of the parameters of the Hamiltonian is widely discussed in \cite{bagbook2}. 
Here this choice is relevant just to allow a simple computation of the eigenvalues and the eigenvectors of $H$, which acts in $\Hil=\mathbb{C}^{16}$ and turns out to be 
the Hermitian matrix:
$$
H=\left[
\begin{array}{cccccccccccccccc}
0 & 0 & 0 & 0 & 0 & 0 & 0 & 0 & 0 & 0 & 0 & 0 & 0 & 0 & 0 & 0 \\
0 & 2 & -1 & 0 & 0 & 0 & 0 & 0 & -1 & 0 & 0 & 0 & 0 & 0 & 0 & 0 \\
0 & -1 & 3 & 0 & -3 & 0 & 0 & 0 & -2 & 0 & 0 & 0 & 0 & 0 & 0 & 0 \\
0 & 0 & 0 & 5 & 0 & -3 & 0 & 0 & 0 & -2 & 1 & 0 & 0 & 0 & 0 & 0 \\
0 & 0 & -3 & 0 & 4 & 0 & 0 & 0 & -3 & 0 & 0 & 0 & 0 & 0 & 0 & 0 \\
0 & 0 & 0 & -3 & 0 & 6 & -1 & 0 & 0 & -3 & 0 & 0 & 1 & 0 & 0 & 0 \\
0 & 0 & 0 & 0 & 0 & -1 & 7 & 0 & 0 & 0 & -3 & 0 & 2 & 0 & 0 & 0 \\
0 & 0 & 0 & 0 & 0 & 0 & 0 & 9 & 0 & 0 & 0 & -3 & 0 & 2 & -1 & 0 \\
0 & -1 & -2 & 0 & -3 & 0 & 0 & 0 & 2 & 0 & 0 & 0 & 0 & 0 & 0 & 0 \\
0 & 0 & 0 & -2 & 0 & -3 & 0 & 0 & 0 & 4 & -1 & 0 & 0 & 0 & 0 & 0 \\
0 & 0 & 0 & 1 & 0 & 0 & -3 & 0 & 0 & -1 & 5 & 0 & -3 & 0 & 0 & 0 \\
0 & 0 & 0 & 0 & 0 & 0 & 0 & -3 & 0 & 0 & 0 & 7 & 0 & -3 & 0 & 0 \\
0 & 0 & 0 & 0 & 0 & 1 & 2 & 0 & 0 & 0 & -3 & 0 & 6 & 0 & 0 & 0 \\
0 & 0 & 0 & 0 & 0 & 0 & 0 & 2 & 0 & 0 & 0 & -3 & 0 & 8 & -1 & 0 \\
0 & 0 & 0 & 0 & 0 & 0 & 0 & -1 & 0 & 0 & 0 & 0 & 0 & -1 & 9 & 0 \\
0 & 0 & 0 & 0 & 0 & 0 & 0 & 0 & 0 & 0 & 0 & 0 & 0 & 0 & 0 & 11 \\
\end{array}
\right].
$$
Using any mathematical software one easily deduces the following list of eigenvalues $\{E_j\}_{j=1}^{16}$, ordered in decreasing absolute value: 
$$
13.6645,\quad 11.4925,\quad11,\quad9.17202,\quad8.82798,\quad7,\quad 6.66453,\quad 6.50749,
$$
$$4.49251,\quad4.33547,\quad4,\quad-2.66453,\quad2.17202,\quad1.82798,\quad-0.492505,\quad0.
$$
Denote by $e_j$ the corresponding normalized eigenvectors for $H$. Then $\F_{e}=\{e_j\}_{j=1}^{16}$ is an ONB of $\Hil$
and $H=\sum_{j=1}^{16}\, E_j|e_j\rangle\langle{e}_j|$. $H$ is the analogous of the Hamiltonian $H_e$ in (\ref{K5b}).

What we are interested in here is to see the mathematical consequences  if we assume that the second level, $L_2$, cannot be filled because of, say, some particular reason: for instance, the organisms in $L_2$ receive nutrients from $L_1$, see Figure \ref{fig01}, but the external conditions can be not sufficiently good to allow these organisms to survive. From a mathematical point of view, this situation can be described by introducing the following orthogonal projector
$$
P_{L_2}=\sum_{i,j,k=0,1}P_{i,j,0,k}, 
$$
where $P_{i,j,0,k}f=\langle \mathsf{e}_{i,j,0,k},f\rangle \mathsf{e}_{i,j,0,k}$, for all vector $f$ in $\Hil$. 
The physical Hamiltonian $H_{ph}=P_{L_2}HP_{L_2}$ is an operator acting in the subspace
$\K=\mbox{span}\{\delta_1, \delta_2, \delta_3, \delta_4, \delta_9, \delta_{10}, \delta_{11}, \delta_{12}\}$ of $\Hil$, where $\{\delta_j\}$ is the canonical ONB of  
$\mathbb{C}^{16}$:
$$
H_{ph}=\left[
\begin{array}{cccccccccccccccc}
0 & 0 & 0 & 0 & 0 & 0 & 0 & 0 & 0 & 0 & 0 & 0 & 0 & 0 & 0 & 0 \\
0 & 2 & -1 & 0 & 0 & 0 & 0 & 0 & -1 & 0 & 0 & 0 & 0 & 0 & 0 & 0 \\
0 & -1 & 3 & 0 & 0 & 0 & 0 & 0 & -2 & 0 & 0 & 0 & 0 & 0 & 0 & 0 \\
0 & 0 & 0 & 5 & 0 & 0 & 0 & 0 & 0 & -2 & 1 & 0 & 0 & 0 & 0 & 0 \\
0 & 0 & 0 & 0 & 0 & 0 & 0 & 0 & 0 & 0 & 0 & 0 & 0 & 0 & 0 & 0 \\
0 & 0 & 0 & 0 & 0 & 0 & 0 & 0 & 0 & 0 & 0 & 0 & 0 & 0 & 0 & 0 \\
0 & 0 & 0 & 0 & 0 & 0 & 0 & 0 & 0 & 0 & 0 & 0 & 0 & 0 & 0 & 0 \\
0 & 0 & 0 & 0 & 0 & 0 & 0 & 0 & 0 & 0 & 0 & 0 & 0 & 0 & 0 & 0 \\
0 & -1 & -2 & 0 & 0 & 0 & 0 & 0 & 2 & 0 & 0 & 0 & 0 & 0 & 0 & 0 \\
0 & 0 & 0 & -2 & 0 & 0 & 0 & 0 & 0 & 4 & -1 & 0 & 0 & 0 & 0 & 0 \\
0 & 0 & 0 & 1 & 0 & 0 & 0 & 0 & 0 & -1 & 5 & 0 & 0 & 0 & 0 & 0 \\
0 & 0 & 0 & 0 & 0 & 0 & 0 & 0 & 0 & 0 & 0 & 7 & 0 & 0 & 0 & 0 \\
0 & 0 & 0 & 0 & 0 & 0 & 0 & 0 & 0 & 0 & 0 & 0 & 0 & 0 & 0 & 0 \\
0 & 0 & 0 & 0 & 0 & 0 & 0 & 0 & 0 & 0 & 0 & 0 & 0 & 0 & 0 & 0 \\
0 & 0 & 0 & 0 & 0 & 0 & 0 & 0 & 0 & 0 & 0 & 0 & 0 & 0 & 0 & 0 \\
0 & 0 & 0 & 0 & 0 & 0 & 0 & 0 & 0 & 0 & 0 & 0 & 0 & 0 & 0 & 0 \\
\end{array}
\right],
$$
whose eigenvalues are the following:
$$
7.38849,\quad7,\quad4.57577,\quad4.18728,\quad2.81272,\quad2.42423,\quad-0.38849, \quad 0.
$$
We see that the eigenvalues of $H$ and $H_{ph}$ are different, and that the spectrum $\sigma(H_{ph})$
is contained in $[E_{min}, E_{max}]$, as stated before. 

 Finally, the set $\F_\varphi=\{\varphi_j=P_{L_2}e_j\}_{j=1}^{16}$  is a PF
in the Hilbert space $\K$ and $$H_{ph}=\sum_{j=1}^{16}\, E_j|\varphi_j\rangle\langle{\varphi}_j|=\sum_{j=1}^{8}\, \tilde E_j|\tilde\varphi_j\rangle\langle{\tilde\varphi}_j|,$$ where $\tilde E_j$ are the eigenvalues above and $\tilde \varphi_j$ are the related eigenvectors, whose set is an ONB in $\Hil_{ph}=\mathbb{C}^8$. As in the previous example, the difference between $E_j$ and $\tilde E_j$ may have consequences on the dynamics of the system, but this aspect will not be considered in this paper.

\section{ $E$-connection between PFs and ONBs}\label{sII.3}
The examples above  suggests to introduce the following, physically motivated, relation between PFs and ONBs in a given Hilbert space $\K$:

\begin{defn}\label{K71} Given a PF $\F_\varphi=\{\varphi_j\in\K, \,j\in \mathbb{J}\}$ and a  bounded set of real numbers $E=\{E_j, \,j\in \mathbb{J}\}$, we say that 
the pair $(\F_\varphi, E)$  is $E$-connected to an ONB $\F_{\tilde{e}}=\{\tilde{e}_k\in\K, k\in\mathbb{J}'\}$ 
if a set of real numbers exists, $\tilde{E}=\{\tilde E_k, \ k\in \mathbb{J}'\}$, such that 
\be\label{K24}
{H_\varphi}=\sum_{j\in\mathbb{J}}\,E_j|\varphi_j\rangle\langle\varphi_j|=\sum_{k\in\mathbb{J}'}\,\tilde E_k|\tilde{e}_k\rangle\langle \tilde{e}_k|.
\en 
\end{defn}

This is exactly what we have seen in the examples considered in Section \ref{sectexe}. In other words, (\ref{K24}) can be used to give two different representations
 of the same operator $H_\varphi$ (the physical Hamiltonian, in our examples), one  in terms of a PF, $\F_{\varphi}$, and one, possibly more relevant for its physical interpretation,  in terms of $\F_{\tilde{e}}$.

If a pair $(\F_\varphi, E)$ is $E$-connected to an ONB $\F_{\tilde{e}}$,  then the real numbers $\tilde{E}_k$ in \eqref{K24} are eigenvalues of the self-adjoint operator
 $H_\varphi$,  whereas $\tilde{e}_k$ are the corresponding normalized eigenvectors.
The eigenvalues $\tilde E_k$ can be defined with the use of Theorem \ref{K10} and obviously, they are defined uniquely by the pair $(\F_\varphi, E)$.  
Moreover,  {due to Remark \ref{rem2}, the cardinality $|\mathbb{J}'|$ coincides with the frame potential ${\bf FP}[\F_\varphi]$ 
and it cannot exceed the cardinality $|\mathbb{J}|={\bf FP}[\F_\varphi]+{\bf e}[\F_\varphi]$ of $\mathbb{J}$.}

If $\F_\varphi$ is an ONB in $\K$,  then the pair $(\F_\varphi, E)$ is $E$-connected to the same 
ONB  $\F_\varphi=\F_{\tilde{e}}$ and $E=\tilde{E}$. A slightly more interesting example is the following:  if $\varphi_{2k-1}=\varphi_{2k}=\frac{1}{\sqrt{2}}{\mathsf{e}}_k$, where 
$\F_{\mathsf{e}}=\{\mathsf{e}_k, \ k\in\mathbb{N}\}$ is an ONB of $\K$, then
the set $\F_\varphi=\{\varphi_j,  \  j\in \mathbb{N}\}$ is a PF in $\K$ and the pair $(\F_\varphi, E)$  is $E$-connected to  $\F_{\mathsf{e}}=\F_{\tilde{e}}$.  In this case, 
 $\tilde{E}=\{\tilde E_k=\frac{1}{2}(E_{2k-1}+E_{2k}), \ k\in \mathbb{N}\}$.

\begin{prop}\label{K81}
{If a PF $\F_\varphi=\{\varphi_j,  \, j=1,2,\ldots, N\}$ has the finite potential $M={\bf FP}[\F_\varphi]$, then 
the pair  $(\F_\varphi, E)$  is $E$-connected to an ONB $\F_{\tilde{e}}=\{\tilde{e}_k\in\K, \  k=1,\ldots{M}\}$.} 
{If,  additionally,  $\F_\varphi$ has the finite excess ${\bf e}[\F_\varphi]$ and $E=\{E_j, \,j=1,\ldots{N}\}$ is} 
a sequence of positive quantities, then the quantities $\tilde{E}_k, \,k=1,\ldots{M}$ in \eqref{K24}
satisfy the relations:
$$
\sum_{k=1}^{M}\tilde{E}_k=\sum_{j=1}^NE_j\|\varphi_j\|^2, \quad \sum_{k=1}^{M}\tilde{E}_k^2=\sum_{j=1}^N\sum_{i=1}^NE_jE_i|\langle \varphi_i, \varphi_j\rangle|^2,
 \quad \sum_{k=1}^{M}\tilde{E}_k^2\geq\frac{1}{M}\left(\sum_{j=1}^NE_j\|\varphi_j\|^2\right)^.  
$$
If all $\varphi_j\not=0$ and $0<\tilde{E}_1\leq\tilde{E}_2\leq\ldots\tilde{E}_{M}$, then for all $1\leq{n}\leq{M}$
\begin{equation}\label{K44}
\sum_{k=1}^{n}\tilde{E}_k\geq\sum_{j=1}^nE_j\|\varphi_j\|^2.
\end{equation}
 \end{prop}
\begin{proof}
The self-adjoint operator $H_\varphi=\sum_{j=1}^N\,E_j|\varphi_j\rangle\langle\varphi_j|$ acts in a finite dimensional space $\K$ with
$\dim\K=M$. Hence its spectrum is discrete and the corresponding normalized eigenfunctions $\{\tilde{e}_k\in\K, k=1,\ldots{M}\}$
form an ONB of $\K$.

If all $E_j$ are positive, then the set $\F_{\hat{\varphi}}=\{\hat\varphi_j=\sqrt{E_j}\varphi_j, \,j=1,\ldots{N}\}$ is a frame in $\K$ and the operator $H_\varphi$ turns out to be a frame
operator $S=\sum_{j=1}^N|\hat\varphi_j\rangle\langle\hat\varphi_j|$ of $\F_{\hat{\varphi}}$.  In this case, the first three relations in  Proposition \ref{K81} follow from
\cite[Section 5]{CazLynch} and \cite[Corollary 2.3]{Caz}. The inequality \eqref{K44} follows from \cite[Theorem 6.3]{CazLynch}.
\end{proof}

 Another simple result is given by the following Lemma, where the validity of the resolution of the identity is used both for the ONB and for the PF.

\begin{lemma}\label{K85}
Let $\F_\varphi$ be a PF and let the real numbers $E_j$ be the same for all $j\in\mathbb{J}$ (i.e.,  $E_j=\lambda$ for all  $j\in\mathbb{J}$).
Then  the pair $(\F_\varphi, E)$ is $E$-connected to every ONB  $\F_{\tilde{e}}$ of $\K$  and $E=\tilde{E}=\{\lambda\}$. 
\end{lemma} 
\begin{proof} Due to the characteristic properties of PFs and the condition $E_j=\lambda$ we obtain  $H_\varphi=\sum_{j\in\mathbb{J}}\,E_j|\varphi_j\rangle\langle\varphi_j|=
\lambda\sum_{j\in\mathbb{J}}\,|\varphi_j\rangle\langle\varphi_j|=\lambda{I}$. 
Hence, formula (\ref{K24}) holds  for any arbitrary choice of ONB $\F_{\tilde{e}}$, since
$\sum_{k\in\mathbb{J}'}\tilde E_k\,|\tilde{e}_k\rangle\langle \tilde{e}_k|=\lambda\sum_{k\in\mathbb{J}'}\,|\tilde{e}_k\rangle\langle \tilde{e}_k|=\lambda I$.
\end{proof}

The following result is surely more interesting.
\begin{prop}\label{propK777}
Let $\F_\varphi$ be a PF and let the set of real numbers $E=\{E_j, \,j\in \mathbb{J}\}$ have only one accumulation point $\lambda$.
Then there exist an ONB $\F_{\tilde{e}}$ and a set of numbers $\tilde E_k$ such that  the pair $(\F_\varphi, E)$ is $E$-connected to $\F_{\tilde{e}}$.
\end{prop}
\begin{proof}
Rewrite the left-hand side of \eqref{K24} as follows:
\begin{equation}\label{K999}
H_\varphi=\sum_{j\in\mathbb{J}}\,E_j|\varphi_j\rangle\langle\varphi_j|=\sum_{j\in\mathbb{J}}\,(E_j-\lambda)|\varphi_j\rangle\langle\varphi_j|+{\lambda}\sum_{j\in\mathbb{J}}\,|\varphi_j\rangle\langle\varphi_j|
\end{equation}
Due to Lemma \ref{K85}, ${\lambda}\sum_{j\in\mathbb{J}}\,|\varphi_j\rangle\langle\varphi_j|=\sum_{k\in\mathbb{J}'}\,\lambda|\tilde{e}_k\rangle\langle \tilde{e}_k|$,
where $\{\tilde{e}_k\}_{k\in\mathbb{J}'}$ is an \emph{arbitrary} ONB of $\K$. On the other hand, since $\lim_{j\to\infty}(E_j-\lambda)=0$, 
the operator (see \eqref{K201}) 
$$
 (\mathcal{E}-\lambda{I})\{c_j\}=\{(E_j-\lambda)c_j\}, \qquad  {D}(\mathcal{E})=\ell^2(\mathbb{J}).
$$
 is compact in $\ell^2(\mathbb{J})$, \cite[problem 132]{Halmoshbook}. Therefore, the operator  (see \eqref{K101})
  $$
 (\mathcal{E}_\varphi-\lambda{I})\{c_j\}=\mathcal{P}(\mathcal{E}-\lambda{I})\{c_j\} ,  \qquad  D(\mathcal{E}_\varphi)=\mathcal{R}(\theta_\varphi)
 $$
acting in  the subspace $\mathcal{R}(\theta_\varphi)$ of $\ell^2(\mathbb{J})$ is compact.  This means that 
 $$
 H_\varphi-\lambda{I}=\sum_{j\in\mathbb{J}}\,(E_j-\lambda)|\varphi_j\rangle\langle\varphi_j|
 $$
is self-adjoint and compact in $\K$ (since $H_\varphi-\lambda{I}$ is unitary equivalent to  $\mathcal{E}_\varphi-\lambda{I}$).
Therefore, there exists an ONB $\{\tilde{e}_k\}_{k\in\mathbb{J}'}$ of $\K$ formed by normalized eigenvectors $\tilde{e}_k$:
\begin{equation}\label{K998}
(H_\varphi-\lambda{I})\tilde{e}_k=\mu_k\tilde{e}_k, \qquad k\in\mathbb{J}'
\end{equation}
and $H_\varphi-\lambda{I}=\sum_{j\in\mathbb{J}}\,(E_j-\lambda)|\varphi_j\rangle\langle\varphi_j|=\sum_{k\in\mathbb{J}'}\,\mu_k|\tilde{e}_k\rangle\langle \tilde{e}_k|$.
Substituting the obtained formulas into \eqref{K999}, we obtain 
$$
\sum_{j\in\mathbb{J}}\,E_j|\varphi_j\rangle\langle\varphi_j|=\sum_{k\in\mathbb{J}'}\,(\lambda+\mu_k)|\tilde{e}_k\rangle\langle \tilde{e}_k|,
$$
where  the choice of ONB $\{\tilde{e}_k\}_{k\in\mathbb{J}'}$ is determined by \eqref{K998}. 
\end{proof}
\begin{rem}
The proof of Proposition \ref{propK777} cannot be modified for the case of two or more  accumulation points of $E=\{E_j, \,j\in \mathbb{J}\}$ 
and the problem of checking  $E$-connection in this case is still open.
\end{rem}

\section{Conclusions}\label{secconcl}

This paper is a continuation of a series of contributions on possible definitions of Hamiltonian-like operators in terms of sets of vectors which are not necessarily ONB, 
as in \cite{bell,bit}. In particular, this is relevant when the physical part of the Hilbert space of the system under analysis produces, in a quite natural  way, 
a PF in terms of which the physical part of the Hamiltonian,  $H_{ph}$, can be expanded. With this in mind we have studied some properties of $H_{ph}$, with 
particular focus on its eigenvalues. In this perspective we have introduced the notion of $E$-connection, and we have analyzed it  for bounded operators. 
{ Unbounded Hamiltonians generated by PF's have a lot of curious properties that require a deeper analysis 
 and they will be considered in a forthcoming paper.}

We have also discussed few examples of our general framework, 
two of which directly connected to biological models considered in the literature in recent years. These are both given in terms of fermionic operators. In this paper we have considered mainly the mathematical consequences of our settings. We plan to consider what changes in the dynamics (and, more in general, in the physics) of the systems in a future paper. In this perspective, a possible application is the use of the idea of localized systems as described in \cite{jpg}. This is particularly relevant for quantization procedures.

\section*{Acknowledgements}
{The authors thank the Referees whose remarks  and suggestions
led to improvements in the paper.}
FB was partially supported by the University of Palermo, via the CORI 2017 action, and by the Gruppo Nazionale di Fisica Matematica of Indam.
SK was partially supported by the Faculty of Applied Mathematics AGH UST statutory tasks within subsidy of 
Ministry of Science and Higher Education of Poland.

\end{document}